\documentclass[preprint,11pt]{article}

\usepackage{amssymb,amsfonts,amsmath,amsthm,amscd,dsfont,mathrsfs}
\usepackage{eqsection,graphicx,float,psfrag,epsfig,color}

\newcommand{\f}{\phi_s} 
\newcommand{\1}{\mathds{1}} 
\newcommand{\vp}{\varphi} 
\newcommand{\supp}{\mathrm{supp}}
\newcommand{\mR}{\mathbb{R}} 

\newcommand{\mZ}{\mathbb{Z}}

\newcommand{\pp}{\mathbb{P}}
\newcommand{\E}{\mathbb{E}}
\newcommand{\A}{[0,\alpha^*)}

\newcommand{\e}{\varepsilon}

\newcommand{\iid}{\stackrel{\text{iid}}{\sim}}

\theoremstyle{definition}
\newtheorem{definition}{Definition}
\theoremstyle{plain}
\newtheorem{thm}{Theorem}
\theoremstyle{plain}
\theoremstyle{plain}
\newtheorem{lemma}{Lemma}
\theoremstyle{plain}
\theoremstyle{plain}
\theoremstyle{plain}
\theoremstyle{remark}
\newtheorem{remark}{Remark}
\theoremstyle{plain}

\def\Cble{{\cal C}}
\def\de{{\rm d}}
\def\prob{\mathbb{P}}

\begin{document}

\title{On the concentration of the number of solutions\\
of random satisfiability formulas}

\author{Emmanuel Abbe\thanks{Department of Computer and Communication Sciences, EPFL, Switzerland. Email: emmanuel.abbe@epfl.ch} \;\;\; and\;\;\;
Andrea Montanari\thanks{Department of Electrical Engineering,
and Department of Statistics,
Stanford University. Email: montanari@stanford.edu}}

\date{}

\maketitle


\begin{abstract}
Let $Z(F)$ be the number of solutions of a random $k$-satisfiability 
formula $F$ 
with $n$ variables and clause density $\alpha$. 
Assume that the probability that $F$ is unsatisfiable is 
$O(1/\log(n)^{1+\e})$ for $\e>0$.
We show that (possibly excluding a countable set of 
`exceptional' $\alpha$'s) there exists a non-random function
$\phi(\alpha)$ such that, for any $\delta>0$,
$(1/n)\log Z(F)\in [\phi-\delta,\phi+\delta]$ with high probability. 
In particular, the assumption holds for all $\alpha<1$,
which proves the above concentration claim in the whole 
satisfiability regime of random $2$-SAT. We also extend
these results to a broad class of constraint satisfaction problems.

The proof is based on an interpolation technique from spin-glass
theory, and on an application of Friedgut's theorem on sharp thresholds for graph properties.
\end{abstract}
%
%
\section{Introduction and main results}

Over the last twenty years, a considerable effort
has been devoted to understanding the typical properties
of random $k$-satisfiability ($k$-SAT) instances. 
This line of work was initially motivated by the empirical 
discovery of a striking relation between the
running time of standard solvers and the 
proximity to the `satisfiability threshold'.
Since then, an important motivation for these investigations
has been to develop better heuristics to cope
with hard constraint satisfaction problems.

Significant progress has been made along this path. 
In particular, it was estabilished early on that, indeed, the probability
that a random instance is unsatisfiable increases sharply from $0$ 
to $1$ when the `clause density' (number of clauses per variable)
crosses a critical threshold \cite{Friedgut}. 
This phenomenon is referred to as the satisfiability phase transition
or satisfiability threshold.
The critical density might a priori depend on the number of variables,
although upper and lower bounds are known to match up to a 
multiplicative constant that goes to $1$ when $k$ increases \cite{AchlioetalNature}.

On the other hand, a significantly more detailed 
picture has been conjectured, building on non-rigorous techniques
from statistical physics such as the replica and cavity methods
\cite{Monasson99Nature,MezParZec_science,OurPNAS,MMBook}.
In particular, not only an $n$-independent 
critical density is conjectured to exist,
but explicit values (depending on $k$) were computed in 
\cite{MezParZec_science}. This type of analysis also lead to
an intriguing picture of the geometry of the set of solutions
of a random satisfiability instance
\cite{MezParZec_science,OurPNAS,MMBook}. While only a small 
subset of these results have been estabilished rigorously,
they provided guidance and stimulus for exciting rigorous
developments \cite{CojaAchlioptas}.

In this paper we explore the most basic property 
of the set of solutions of a random $k$-SAT instance:
its cardinality. 
The problem of computing the number of solution of a $k$-satisfiability 
($k$-SAT) formula is well known to be $\#$P-complete for 
any $k\ge 2$. Even worse, there is no fully polynomial 
randomized approximation scheme (FPRAS) to approximapte
the number of solutions  unless NP = RP \cite{DyerComplexity}. 

Here we are instead interested in asymptotic estimates of the number of 
solutions of a random formula, in the following sense. 
Let $F(n,\alpha)$ denote a formula on 
$n$ variables with clause density $\alpha$ (each clause having $k$ literals) and let
$Z(F(n,\alpha))$ be the number of solution of $F(n,\alpha)$. 
From a statistical physics viewpoint, one of the most basic conjectures 
is that $Z(F(n,\alpha))$  concentrates on the exponential scale.
Namely, for each $\alpha<\alpha_{\rm s}(k)$ 
(the satisfiability threshold), there exists $\phi = \phi(\alpha)$ non-random
such that, for any $\delta>0$,
$2^{n(\phi-\delta)}\le Z(F(n,\alpha)) < 2^{n(\phi+\delta)}$ with high probability.
In formula
\begin{eqnarray}
\lim_{n\to\infty}
\pp\{2^{n(\phi-\delta)}\le Z(F(n,\alpha)) < 2^{n(\phi+\delta)}\} = 1\, .
\label{eq:Conjecture}
\end{eqnarray}
While for arbitrary $k \geq 3$ we cannot establish this conjecture for any $\alpha$ below the satisfiable threshold (which would imply
in particular the existence of an $n$-independent
satisfiability threshold), we are able to prove
that (\ref{eq:Conjecture}) holds for any $\alpha$ such that 
 $\pp\{Z(F(n,\alpha))=0\}$ (the unsatisfiability probability)
is upper bounded by $1/(\log n)^{1+\e}$ for some $\e>0$ and all $n$ large enough. 
In particular, this establishes the 
conjecture for $k=2$ (in the entire satisfiable phase), and for $k\ge 3$ and $\alpha<1$. 
A verification of previous assumption for arbitrary values of $k$ and $\alpha$ up to the so-called `dynamic threshold' is also discussed. 

We further generalize these results to a broad family
of constraint satisfaction problems.

One basic difficulty in estabilishing (\ref{eq:Conjecture})
is that the concentration of $\log Z(F)$ cannot be proved using 
standard martingale methods. Such an argument typically requires
to control the difference $|\log Z(F')-\log Z(F)|$ for
$F$ and $F'$ differing in a single clause \cite{Ledoux}.
Unfortunately, adding a single clause can change the value of
$Z(F)$ from exponentially large to $Z(F') = 0$.

In order to overcome this difficulty, we use Friedgut's theorem to
prove that the property $Z(F)\le 2^{n\phi}$ has a sharp threshold.
We then translate this result into a concentration statement by proving
that $(1/n)\E\log(1+Z(F))$ converges, using an interpolation technique.
This is where the condition on the decay of the unsatisfiability probability is used, since the interpolation technique does not lead to a formal superadditivity property in this setting. 
 
The interpolation method we use was first developed by
Francesco Guerra and Fabio Toninelli in the context of spin glass theory
\cite{GuerraToninelliLimit} and then generalized to a number of problems
from statistical physics, computer science
\cite{FranzLeone,FranzLeoneToninelli,PanchenkoTalagrand}
and coding theory \cite{Montanari05,Kudekar}.
A nice suite of combinatorics applications were
recently presented in \cite{BayatiInterpolation}.

In particular, the interpolation technique was applied to random $k$-SAT
in \cite{FranzLeone,FranzLeoneToninelli}, to show that a suitable normalized log-partition function
has a limit for $n\to\infty$.
The same work implies that the minimum number of unsatisfied 
clauses per variable has a limit as well. 
Let us emphasize a crucial difference between such results and
the conjecture (\ref{eq:Conjecture}): 
The quantities (normalized log-partition function or optimal cost)
considered in  \cite{FranzLeone,FranzLeoneToninelli}
and following-up work are easily proved to
concentrate via martingale methods. This considerably simplifies those proofs.

\vspace{0.2cm}

{\bf Acknowlegements.} 
The present research was initiated as far back as 2006, and remained dormant 
for a long period of time. Some results were presented at the
`DIMACS Working Group on Message-Passing Algorithms' in October 2008.
We were finally motivated to polish and publish the manuscript
after discussions with Mohsen Bayati, David Gamarnik and Prasad Tetali
regarding their recent paper  \cite{BayatiInterpolation}. It is a pleasure 
to thank them. The first author also thanks Emre Telatar for a careful reading
of some parts of the manuscript and for useful suggestions. 

This work was partially supported by a Terman fellowship,
the NSF CAREER award CCF-0743978 and the NSF grants DMS-0806211,
CCF-0915145.
%
%
\section{Random $k$-SAT}

\begin{definition}
A $k$-clause is a disjunction of $k$ Boolean variables or their negations. 
Let $C_k(n)$ be the set of all $ N={n \choose k} 2^k $ possible 
$k$-clauses on $n$ Boolean variables. 
We denote by $F_k(n,\alpha)$ a random formula which is formed by selecting 
independently each element in $C_k(n)$ with probability $p_k(n,\alpha)=
\alpha n/({n \choose k} 2^k)$, and by taking the conjunction of 
the selected clauses. 
\end{definition}

The number of clauses in the above model is a binomial random variable,
which concentrates exponentially fast around its expectation $\alpha n$.
Some of our computations prove to be simpler within slightly
different models, whereby the number of clauses is either
Poisson or deterministic with the same mean  $\alpha n$. 
Standard monotonicity arguments can be used to show 
the equivalence of these models for our purposes and we will hence 
switch freely between these different models. 

Unless specified, the value of $k$ will remain fixed throughout the paper 
and the $k$ subscript is dropped. 

We denote by $Z(F)$ the number of satisfying assignments   (solutions) of a 
Boolean formula $F$ and by
\begin{eqnarray*}
P_n(\alpha,\phi) := \pp\{Z(F(n,\alpha)) < 2^{n\phi}\}\, ,
\end{eqnarray*}
the probability that a random formula has less than $2^{n\phi}$ 
satisfying assignments.

\begin{definition}
Note that $P_n(\alpha,0)= 
\pp\{Z(F(n,\alpha)) =0 \}$ is the probability that $F(n,\alpha)$ 
is unsatisfiable (UNSAT). We define
\begin{eqnarray*}
\alpha^* := \sup \big\{ \,\alpha :\; 
P_n(\alpha,0) = O(1/(\log n)^{1+\e}), \mbox{ for some }\e>0 \big\}\, . 
\end{eqnarray*}
\end{definition}

\noindent
{\bf Remark 1.} the results in this paper still hold when defining 
$\alpha^*$ to be $\sup \{ \alpha : \sum_n P_n(\alpha,0)/n 
< \infty\}$. However, 
the above definition simplifies the proofs without much loss of generality.

\noindent
{\bf Remark 2.} For any $k\ge 2$, we have $\alpha^* \geq 1$. 
Indeed, considering the case of $2$-SAT, \cite{DeLaVega}
proves that, for $\alpha < 1$, $\pp \{Z(F_2(n,\alpha)) =0\}=O(1/ n)$.
Since for any $k \geq 3$, $\pp \{Z(F_k(n,\alpha)) =0\} 
\leq \pp \{Z(F_2(n,\alpha)) =0\}$, we conclude 
$$\alpha^* \geq 1$$
for $k\ge 3$ as well. 

Unfortunately, the bounds on the satisfiability threshold
based on the second moment method 
\cite{AchlioetalNature} do not imply any quantitative estimate
on the probability that $F_k(n,\alpha)$ is UNSAT.
It might be possible to prove such an estimate by a careful
analysis of specific solution algorithms. In particular,
a careful analysis of the recent algorithm \cite{Coja09} 
might lead to a proof of 
$\alpha^* \geq 2^k(1-\delta) \log k/k$, for $k$ large enough \cite{COPersonal}. 
We expect that $\alpha^*$ does coincide with the satisfiability 
threshold. 

\vspace{0.2cm}

Our main result estabilishes the conjecture (\ref{eq:Conjecture})
for $\alpha<\alpha^*$, apart possibly for countably many
`special' values of $\alpha$.
\begin{thm}\label{res}
There exist a countable set $\Cble\subseteq[0,\alpha^*)$,
and $\phi_s: [0,\alpha^*)\to [0,1]$
such that the following holds.
For any $\e>0$ and any $\alpha \in [0,\alpha^*)\setminus\Cble$, we have
\begin{eqnarray*} 
\lim_{n \rightarrow \infty} P_n(\alpha,\phi_s(\alpha)- \varepsilon) &=& 0, \\
\lim_{n \rightarrow \infty} P_n(\alpha,\phi_s(\alpha) + \varepsilon) &=& 1.
\end{eqnarray*}
\end{thm}
In order to prove this result, which is done in Section \ref{proof3}, we first show the following sharp 
threshold result. 
\begin{thm}\label{fried}
For any $\varepsilon>0$ and $\phi \in [0,1)$, there exists $\{\alpha_n(\phi)\}_{n \in \mZ_+}$ such that 
\begin{eqnarray*}
\lim_{n \rightarrow \infty}P_n(\alpha_n(\phi)  - \varepsilon, \phi) &=& 0, \\
\lim_{n \rightarrow \infty}P_n(\alpha_n(\phi)+ \varepsilon, \phi) &=& 1.
\end{eqnarray*}
\end{thm}
In other words, for any fixed $\phi$, the property 
$\{ Z(F(n,\alpha)) < 2^{n\phi} \}$ has a sharp threshold in $\alpha$. 
In \cite{Friedgut}, this result is shown for $\phi=0$,
 i.e., for the property $\{F(n,\alpha) 
\text{ is UNSAT} \}$. As for the $\phi=0$ case, Theorem \ref{fried} is proved by showing that the monotone property $\{ Z(F(n,\alpha)) < 2^{n\phi} \}$ cannot be approximated by a ``local property'', and hence must have a sharp threshold. 
For greater convenience of the reader, and in order to introduce
relevant notations, we reproduce in Section \ref{sec:Friedgut} the Friedgut's Theorem on sharp thresholds for monotone properties
\cite{Friedgut}. The proof of Theorem \ref{fried} is differed to Section \ref{prooffried}.


%
%
In order to prove Theorem \ref{res}, we then would like 
to transfer the threshold in $\alpha$ (Theorem \ref{fried}) into 
a threshold in $\phi$. This step is however not straightforward because of the little knowledge we have about the function $\phi \mapsto \alpha_n(\phi)$. 
In order to establish this threshold transfer, 
we first prove the following result, which shows  the existence 
of the  $n\to\infty$ limit of the normalized logarithm of
the number of solutions when $\alpha < \alpha^*$. 
 
\begin{thm}\label{convpsat}
Let $$ \psi_n (\alpha) :=  \frac{1}{n} \E [ \log Z(F(n,\alpha)) | Z(F(n,\alpha)) \geq 1 ] .$$
We have that $\psi_n (\alpha)$ converges to a limit $\phi_s(\alpha)$, for every $\alpha < \alpha^*$.
\end{thm}
This theorem is proved in Section \ref{proof2} and the key step in the proof consists in establishing the following pseudo-superadditivity property.
\begin{lemma}\label{guerra}
For any $\alpha$, let $Z_{n}:= Z(F(n,\alpha))$, we then have for any $n_1,n_2 \geq k$
\begin{align*}
 \E  \log (1+Z_{n_1+n_2}) \geq  \E \log (1+Z_{n_1} Z_{n_2})\, .
\end{align*}
\end{lemma}

The proof of Lemma \ref{guerra} is differed to Section \ref{proof2} and is based on the interpolation technique by 
Guerra and Toninelli \cite{GuerraToninelliLimit}, and Franz-Leone
\cite{FranzLeone}. However, while in those cases 
one obtains superadditivity  of $\E\log Z$, in the present case we
get a weaker result because of the ``1+'' term. 
Hence, in this case the existence of the infinite volume 
limit is not a straightforward consequence of Lemma \ref{guerra}. 
Notice that this problem is intimately related to the fact that
$Z=0$ with positive probability, and therefore $\E\log Z$ 
is not defined. We use the fact that $\alpha < \alpha^*$ to circumvent this problem in Theorem \ref{convpsat}. %

Finally, Theorem \ref{convpsat} implies the following result, which allows the complete the threshold transfer. 
\begin{lemma}\label{conva}
For each $n$ and $\phi$, let 
$\alpha_n (\phi)$ be such that $P_n(\phi,\alpha_n (\phi)) =1/2$. 
We have that $\alpha_n(\phi)$ converges for almost every $\phi \in \phi_s(\A)$ (where $\phi_s$ is defined in Theorem \ref{convpsat}). 
\end{lemma}
Lemma \ref{conva} is proved in Section \ref{proof3}.

%

\section{A general family of random CSP's}

In this section, we extend the results of the previous section to a general 
family of random constraint satisfaction problems (CSP) over binary variables. 
An ensemble in this family is defined as follows.
\begin{definition}
 Let $\mu$ be a distribution over Boolean functions $\varphi:\{-1,+1 \}^k\to\{0,1\}$, which we call the 
{\it clause type distribution}. Let $n$ be an integer and $\alpha \in \mR_+$. 
A random formula from the ensemble $F_k(\alpha, n, \mu)$
is drawn as follows. For each $a\in\{1,\dots,m=\lfloor \alpha n \rfloor\}$
the $a$-th clause is drawn independently from previous
ones. For  clause $a$, $k$ indices $i_1(a),\ldots,i_k(a)$ are drawn 
independently and uniformly at random in $[n]$. Further 
$\vp_a: \{-1,+1 \}^k \rightarrow \{0,1\}$ is drawn under the distribution 
$\mu$, producing the clause $\vp_a(x_{i_1(a)}, \ldots,x_{i_k(a)})$. 

An assignment $x\in\{+1,-1\}^n$ is said to 
\emph{satisfy} the formula $F_k(\alpha, n, \mu)$
if, for each $a\in [m]$,   we have $\vp_a(x_{i_1(a)}, \ldots,x_{i_k(a)})=1$. 
\end{definition}

As in previous section, we will often drop the subscripts $k$ in the following.  
Further, $Z (F)$ is the number of satisfying assignments of formula $F$
and we define
\begin{eqnarray*}
P_n(\alpha, \phi, \mu) := \pp\{ Z(F(\alpha, n, \mu)) < 2^{n\phi} \}\, .
\end{eqnarray*}

\begin{definition}
Note that $P_n(\alpha,0, \mu)=
 \pp\{Z(F(\alpha,n, \mu)) =0 \}$ is the probability that $F(\alpha,n, \mu)$
is UNSAT. We define
\begin{eqnarray*}
\alpha^*(\mu) := \sup \big\{ \alpha :\;
 P_n(\alpha,0, \mu) = O(1/(\log n)^{1+\e})
\mbox{ for some }  \e>0 \big\}\, .
\end{eqnarray*}
\end{definition}

\begin{definition}
For $\vp: \{-1,1\}^k \rightarrow \{0,1\}$ and $\theta \in [-1,1]$, let 
$$\| \vp \|_\theta^2 = \sum_{x \in \{-1,1\}^k} \vp(x)^2 v_\theta (x)\quad 
\mbox{ and }\quad \| \vp \| = \| \vp \|_0$$
where $$v_\theta (x) = \prod_{i=1}^k \frac{1+x_i\theta}{2}.$$
\end{definition}
Note that $\| \vp \|_\theta^2$ is the probability that $\vp=1$ under the measure $v_\theta$, which assigns probabilities $(1-\theta)/2$ and $(1+\theta)/2$ to $-1$ and $+1$ respectively.

Our CSP ensemble is specified by the distribution $\mu$ over clause types and 
we now describe two set of hypotheses on this distribution.

\vspace{0.2cm}

\noindent
{\bf H1.} (a) \emph{Dominance of balanced assignments.} For every 
$\theta \in [-1,1]$,
$\E_\vp \log \| \vp \|_\theta \leq \E_\vp \log \| \vp \|,$
with equality only if $\theta=0$. This condition implies that, in a typical random instance, most solutions have almost as many $1$'s as $-1$'s.\\
(b) \emph{Unsatisfiability of uniform assignments.} For every $s\in\{-1,+1\}$, 
there is at least one clause $\vp$ with $\mu(\vp)>0$ such that $\vp 
(s, \ldots, s)=0$.

\vspace{0.2cm}

\noindent{\bf H2.} \emph{Convexity of $\Gamma_l$.} Let $M_1(\{-1,1\}^l)$ be the set of probability measure on $\{-1,1\}^l$ and let
\begin{align*}
\Gamma_l: \, M_1(\{-1,1\}^l) &\to \mR\, 
\end{align*}
be defined by
\begin{align}
 \Gamma_l(\nu):=
\E_\vp \E_{Z^{(r)}} \prod_{r=1}^l (1-\vp(Z^{(r)})) \label{cond1}
\end{align}
where $Z^{(r)}$ are Boolean random vectors of dimension $k$ such that $Z_i=(Z_i^{(1)},\ldots, Z_i^{(l)})$, $i=1,\ldots,k$, are i.i.d. with distribution $\nu$, whereas $\vp$ is a random clause type drawn under distribution $\mu$. 
Under H2, $\mu$ is assumed to make $\Gamma_l$ convex for for any $l \geq 1$.

\vspace{0.2cm}

Notice that conditions H1.(a) and H1.(b) coincide with  conditions 4 and 5 
in \cite{MRT09}.
Further, hypothesis H1 is satisfied by a number of interesting
random CSP ensembles. Such examples include
\begin{itemize}
\item $k$-NAE-SAT, where $\vp(x)=\vp_s(x)= \1(x \notin \{-s,s\})$ and $\mu(\vp_s)=2^{-k}$ for each $s \in \{-1,1\}^k$;
\item Hypergraph $2$-coloring, where
$\vp(x)= \1(x \notin \{\underline{-1};
\underline{+1}\})$ is the unique clause in the support of $\mu$, with $\underline{-1} = (-1,\dots,-1)$
and  $\underline{+1} = (+1,\dots,+1)$;
\item $k$-XOR-SAT, where $\vp(x) = \vp_s(x) = \1(\prod_{i=1}^k x_i =s)$ and 
$\mu(\vp_s)=1/2$ for each $s\in \{-1,1\}$;
\item $k$-SAT, where $\vp(x)=\vp_s(x)= \1(x \neq s)$ and $\mu(\vp_s)=2^{-k}$ for each $s \in \{-1,1\}^k$.
\end{itemize}
For the first three examples above, it is checked in 
\cite{MRT09} that hypothesis H1 is satisfied.
 Let us check that this is the case for $k$-SAT as well. Note that 
\begin{align*}
\E_s \| \vp \|_\theta^2= 1-\E_s  \prod_{i=1}^k \frac{1-s_i \theta}{2} =1-2^{-k} = \E_s \| \vp \|^2, 
\end{align*}
hence
\begin{align*}
\E_s \log \| \vp \|_\theta^2 \leq \log \E_s \| \vp \|_\theta^2 =\log \E_s \| \vp \|^2 =  \E_s  \log\| \vp \|^2.  
\end{align*}
This verifies condition H1.(a). Condition H1.(b) holds trivially.

Hypothesis H2 is not straightforward to check.
The next definition characterizes a family of clause type distributions 
satisfying it.
\begin{definition}[$k$-factorizing distributions]
A clause type distribution $\mu$ is said to {\it k-factorize} if it has the following structure. There exists an integer $J \geq 1$, 
such that any $\vp \in \supp (\mu)$ is of the form
\begin{align}
\vp (x) = \1 (x \notin \{s^{(1)}, \ldots, s^{(J)}\}), 
\end{align}
for some $s^{(1)}, \ldots, s^{(J)} \in \{-1,1\}^k$,
and 
\begin{align}
\mu ( \vp) = \prod_{i=1}^k \bar{\mu} (s_i^{(1)}, \ldots, s_i^{(J)}) \label{cond2}
\end{align}
where $\bar{\mu}$ is a probability distribution on $\{-1,1\}^J$.
In other words, the vectors $(s_i^{(1)}, \ldots, s_i^{(J)})$, for $i=1,\ldots,k$, can have correlated components but are mutually i.i.d. with distribution $\bar{\mu}$.
\end{definition}
This definition can be generalized by letting  $J$ 
itself to be random, but we stick to the above case for the sake of simplicity.

The class of $k$-factorizing clause type distributions includes, among other problems: 
\begin{itemize}
\item $k$-NAE-SAT: $\vp(x) = \1 (x \notin \{-s,s \})$ 
for $s\in\{-1,+1\}^k$ uniformly random. This is $k$-factorizing 
with $\bar{\mu}(-1,1)  = \bar{\mu}(1,-1)=1/2$;
\item Hypergraph $2$-coloring: $\vp(x) = \1 (x \notin \{\underline{-1},\underline{+1}\})$ with $\bar{\mu}(-1,1)  = 1$; 
\item $k$-SAT: $\vp(x) = \1 (x \notin \{s\})$ with $\bar{\mu} (1) = \bar{\mu}(-1) = 1/2$.
\end{itemize}
Condition H2 is satisfied by $k$-factorizing distributions as stated formally
below.
\begin{lemma}\label{ch1}
The mapping $\Gamma_l$ is convex for any $l \geq 1$ if the clause type distribution k-factorizes.  
\end{lemma}
Note that $k$-XOR-SAT does not belong to this class of distributions, nevertheless, condition H2 holds in this case as well, as stated below.
\begin{lemma}\label{ch2}
The mapping $\Gamma_l$ is convex for any $l \geq 1$ for k-XOR-SAT with k even.  
\end{lemma}
The proofs of Lemma \ref{ch1} and Lemma \ref{ch2} are differed to Section \ref{proof4}.
We now state the equivalent of Theorem \ref{res} for this general class of 
CSPs.
\begin{thm}\label{res2}
Assume $\mu$ to satisfy conditions H1 and H2. Then there exists a countable 
set $\Cble$ and a function $\alpha\mapsto \phi_s(\alpha)$ 
such that, for any  $\alpha \in [0,\alpha^* (\mu))\setminus\Cble$, 
and  any $\e>0$, 
\begin{eqnarray*} 
\lim_{n \rightarrow \infty} P_n(\alpha,\phi_s(\alpha)- \varepsilon, \mu) &=& 0, \\
\lim_{n \rightarrow \infty} P_n(\alpha,\phi_s(\alpha) + \varepsilon, \mu) &=& 1.
\end{eqnarray*}
\end{thm}

As in previous section, the proof of this theorem is based on the following 
two theorems. 
\begin{thm}
For any $\mu$ satisfying H1 and $\phi \in [0,1)$, 
there exists $\{\alpha_n(\phi)\}_{n \in \mZ_+}$ such that for any $\e>0$,
\begin{eqnarray*}
\lim_{n \rightarrow \infty}P_n(\alpha_n(\phi)  - \varepsilon, \phi, \mu) &=& 0, \\
\lim_{n \rightarrow \infty}P_n(\alpha_n(\phi)+ \varepsilon, \phi, \mu) &=& 1.
\end{eqnarray*}
\end{thm}
This theorem is proved in \cite{MRT09}. Although in that paper 
a larger set of conditions was required in the formal statement,
by simple inspection of the proof it follows that only conditions 
H1.(a) and H1.(b) were in fact used.

\begin{thm}\label{convcsp}
Let $$ \psi_n (\alpha) :=  \frac{1}{n} \E [ \log Z(F(\alpha,n, \mu)) | Z(F(\alpha,n, \mu)) \geq 1 ] .$$
For any $\mu$ satisfying H2 and for any $\alpha < \alpha^*(\mu)$,
$\psi_n (\alpha)$ converges to a limit $\phi_s(\alpha)$.
\end{thm}
The proof of this 
theorem is based on the following pseudo-superaddittivity lemma.
\begin{lemma}\label{guerracsp}
For any $\alpha$ and $\mu$ satisfying H2, let $Z_n:=Z(F(\alpha,n, \mu))$, we then have for any $n_1,n_2 \geq k$, 
\begin{align*}
 \E  \log (1+Z_{n_1 + n_2}) \geq  \E \log (1+Z_{n_1} Z_{n_2}).
\end{align*}
\end{lemma}
The proof of Lemma \ref{guerracsp} is differed to Section \ref{proof4}. The proofs of Theorem \ref{convcsp} and Theorem \ref{res2} follow the same analytical arguments as for the $k$-SAT case and hence we omit the details here and refer to the proofs of Theorem \ref{convpsat} and Theorem \ref{res}.
%
%

\section{Proofs}
%
%
\subsection{Friedgut's theorem: A reminder}
\label{sec:Friedgut}

Let $\prob_p$ be the product measure on 
$\{0,1\}^N$, such that for $(z_1,\ldots,z_N) \in \{0,1\}^N$, 
$\prob_p (z_1,\ldots,z_N) = p^{\sum_{i=1}^N z_i}  (1-p)^{N - 
\sum_{i=1}^N z_i}$.
If $N= \binom{n}{k}2^k$, the space $\{0,1\}^N$ is isomorphic
to the space of $k$-CNF formulas on $n$ variables and the measure $\pp_p$ is the measure used in this paper for random $k$-SAT.
We say that $A \subseteq \{0,1\}^N$ is monotone if, for any 
$x \in \{0,1\}^N, y \in A$, $x \leq y$ (component-wise) implies 
$y \in A$. Such a collection of vectors $A$ is also called a monotone 
property, and note that by monotone it is really meant here monotonically 
increasing.
When $N= \binom{n}{k}2^k$, a property is called \emph{symmetric}, equivalently called a formula property, if it is invariant under $k$-CNF 
formula automorphisms (invariant under the action of the wreath product of the symmetric group $S_n$ with $k$ copies of $\mZ_2$). 
For instance, the UNSAT property of a  $k$-CNF formula
is monotone and symmetric.

Note that for $N$ fixed and a monotone symmetric property $A=A_N$, the 
function $p \mapsto \pp_p (A)$ is a monotonically increasing 
polynomial with $\pp_0 (A)=0$ and $\pp_1 (A)=1$
(strictly increasing if $A\neq \{0,1\}^N$).
We can therefore define for 
$\e \in [0,1]$, $p_\e$ such that $\pp_{p_\e} (A) = \e$ and we call 
$\delta(\e) = p_{1-\e} - p_\e$ the {\it critical interval}. 
We then say that the property $A$ (or the sequence of properties $A_N$) 
has a {\it sharp threshold} if for all $\e \in (0,1)$, $\delta(\e)/p_{1/2}$ 
tends to zero when $N$ increases. 
Note that for a monotone property $A$ having a 
sharp threshold, there exists $\hat{p}\in(p_{\e},p_{1-\e})$
such that $\pp_p(A) \rightarrow 0$ if $p \leq \hat{p}(1-\e)$ and 
$\pp_p(A) \rightarrow 1$ if $p \geq \hat{p}(1+\e)$. 

If instead, for $\e \in (0,1)$, 
$\delta(\e)/p_{1/2}$ is bounded away from zero, we say that $A$ has a 
{\it coarse threshold}. 
In typical examples, $\delta(\e)/p_{1/2}$ is bounded and hence admits
subsequential limits. In order to prove a sharp threshold for monotone properties,
it is therefore sufficient to rule out the case of a coarse threshold
(modulo reducing to subsequences).
If $A$ has a coarse threshold, then there exists $p^*\in(p_{\e},p_{1-\e})$ 
such that  $p^* \cdot  \frac{\partial}{\partial p} \pp_p (A)\big|_{p^*} 
\leq 1/C$ uniformly in $N$.

For a given formula $F$, we denote 
by $|F|$ the number of clauses in $F$. The average degree of a formula $F$ is the ratio between the number of variables and the number of clauses in $F$, and a balanced formula is a formula whose average degree is no less than that of any sub-formula. We also denote by $E(F)$ the expected number of sub-formulas isomorphic to $F$ in a random formula. 
%
We now state the result of Friedgut.
\begin{thm}[\cite{Friedgut}]\label{thm:GeneralSharp}
Let $0<\beta<1$. There exist functions $B=B(\epsilon,c)$, $b_1=b_1(\epsilon,c)$ and $b_2=b_2(\epsilon,c)$ such that for any $N$, $c>0$, $\e>0$, $p$ and any monotone symmetric family $A$ of $k$-CNF formulas with $n$ variables such that 
$p \cdot \frac{\partial}{\partial q} \pp_q (A)\big|_p \leq c$ and $\beta < \pp_p(A) <1-\beta$, 
there exists a formula $G$ satisfying:
\begin{itemize}
\item $G$ is balanced 
\item $b_1 < E (G) < b_2$ 
\item $|G| \leq B$ 
\item $\pp_p \{ A|G \} \geq 1-\epsilon$, where $\pp_p \{ A|G \}$ denotes the probability that a random formula belongs to $A$ conditioned on the appearance of a specific copy of $G$ in the random formula. 
\end{itemize}
\end{thm}
Note that conditioning of the appearance of a formula satisfying the above hypothesis is not the same as conditioning on a specific copy of such a formula (with specified variables).

%
%
\subsection{Proof of Theorem \ref{fried}}\label{prooffried}

In order to prove Theorem \ref{fried}, let us assume that the property $\{ Z(F(n,\alpha)) < 2^{n \phi} \}$ has a coarse threshold. Let us consider $p_{\rm s}$ in the critical interval of this property such that $p_{\rm s} \cdot \frac{\partial}{\partial p} \pp_p (A)\big|_{p_{\rm s}} \leq c$
and define $p_{\rm s}=\alpha_{\rm s} n/N$.
By compactness there exists a subsequence of size $n$ 
along which $P_n(\alpha_{\rm s},\phi )$ has a limit $P_{\infty}=[\e,1-\e]$.
For simplicity of notation, consider $P_{\infty}=1/2$ (the
general case is completely analogous). 
We can hence assume w.l.o.g. that $\alpha_{\rm s}=\alpha_{\rm s}(n)$ is such 
that $ P_n(\alpha_{\rm s},\phi )= 1/2$ (and $\alpha_{\rm s}$ is well defined since $P_n$ is continuous and strictly increasing in $\alpha$).

The proof consists in showing that there does not exists a specific ``short'' formula $G$ of the kind described in Theorem \ref{thm:GeneralSharp}, providing a contradiction with the coarse threshold assumption. 
We proceed by showing that the third assumption on $G$, the bounded number of clauses, leads directly to a contraction with the fourth one. 
This is shown with the following lemma.

\begin{lemma}\label{crit}
Let $G$ be a satisfiable formula with a finite number of variables and clauses. For any $\e>0$, we have for $n$ large enough
\begin{eqnarray*}
\pp \{ Z(F(n,\alpha_{\rm s})) < 2^{n \phi} | G\} &\leq& 1/2 +\e\, ,
\end{eqnarray*} 
where $\pp \{\,\cdot\, | G\}$ is interpreted as in the
statement of Theorem \ref{thm:GeneralSharp}.
\end{lemma}


Before proceeding to the proof of Lemma \ref{crit}, we need the following technical estimates for the case $\phi >0$. 
\begin{lemma}\label{lip1} 
For any $n, \alpha, \phi, \psi$, with $\phi \geq \psi$, 
\begin{align*}
& \text{(i)} \quad \frac{\partial\phantom{\alpha}}{\partial \alpha} P_n( \alpha, \phi) \leq C_1 \sqrt{n} ,\\
& \text{(ii)} \quad 
P_n(\alpha, \phi) - P_n(\alpha, \psi) \leq C_2(\phi -\psi) \sqrt{n} \log(n) + o(1),
\end{align*}
where $C_1$ and $C_2$ do not depend on $n$.
\end{lemma}
\begin{proof} 
$(i)$ By definition
\begin{align*}
P_n(\alpha, \phi)
=  \sum_{0 \leq m \leq N} {N \choose m} p^m (1-p)^{N-m} \pp \{ Z(F(n,\alpha)) <  2^{n \phi} | |F|=m\},
\end{align*}
and $p=(n\alpha)/N$, thus
\begin{align*}
\frac{\partial}{\partial \alpha}   p^m (1-p)^{N-m}
= \frac{n}{N}   p^{m-1} (1-p)^{N-m-1}(m-Np)
\end{align*}
and
\begin{eqnarray*}
\frac{\partial}{\partial \alpha} P_n(\alpha, \phi) &=&  \frac{n}{N}   p^{-1} (1-p)^{-1}\\
&&\sum_{0 \leq m \leq N} {N \choose m} p^{m} (1-p)^{N-m}(m-Np) \pp \{ Z(F(n,\alpha)) <  2^{n \phi} | |F|=m\}\\
&\leq&  \frac{n}{N}   p^{-1} (1-p)^{-1} \E[(|F|-Np),|F|\geq Np]\\
&\leq&   \frac{n}{N}   p^{-1} (1-p)^{-1} (\E(|F|-Np)^2)^{1/2}\\
&=&   \frac{n}{N}   p^{-1} (1-p)^{-1} (Np(1-p))^{1/2}\\
&\leq& C_1 \sqrt{n},
\end{eqnarray*}

$(ii)$ Note that because of Lemma \ref{lip1} $(i)$, 
$$ \alpha \mapsto \pp \{ Z(F(n,\alpha)) < 2^{m} \}$$ is Lipschitz with constant $C_1 \sqrt{n}$. Therefore, it is sufficient to show that for any $a \geq 0$
\begin{eqnarray}
 \pp \{ Z(F(n,\alpha)) < 2^{a+1} \} \leq  \pp \{ Z(F(n,\alpha+A\log(n)/n)) < 2^{a}\}  + c \log(n)/\sqrt{n}, \label{l}
 \end{eqnarray}
for some constants $c$ and $A$ (possibly depending on $k$), since then 
\begin{eqnarray*}
 &&\pp \{ Z(F(n,\alpha)) < 2^{a+1} \} - \pp \{ Z(F(n,\alpha)) < 2^{a} \} \\
 &\leq&  \pp \{ Z(F(n,\alpha+A\log(n)/n)) < 2^{a}\}- \pp \{ Z(F(n,\alpha)) < 2^{a} \}  + c\log(n)/\sqrt{n}\\
 &\leq& C_2\log(n)/\sqrt{n},
\end{eqnarray*}
for some constant $C_2$.
We first verify \eqref{l} for the model where the formulas have exactly $\alpha n$ clauses drawn uniformly at random, using in that case the notation $F_{\alpha n}$. The result for the Binomial model follows then from standard monotonicity arguments.  
We will show that there exists $\theta \in (0,1)$ such that for any $a\geq 0$ and any integer $l\geq 1$,
\begin{eqnarray}
 \pp \{ Z(F_{\alpha n}) < 2^{a+1} \} -\pp \{ Z(F_{\alpha n+l}) < 2^{a}\} \leq 2 \theta^l. \label{finite}
 \end{eqnarray}
Note that for a deterministic formula $D$, and $C$ a $k$-clause uniformly drawn in $C_k(n)$,
$$\E_C Z( D \wedge C) = (1-\frac{1}{2^k}) Z (D),$$
hence for $l$ uniformly drawn $k$-clauses 
$$\E_{C_1,\ldots,C_l} Z(D \wedge C_1 \wedge \ldots \wedge C_l) \leq (1-\frac{1}{2^k})^l Z (D).$$
Therefore, denoting by $\E_{l}= \E_{C_1,\dots,C_l}$ the expectation with 
respect to $C_1,\dots,C_l$, we get
$$ \pp \{ Z(F_{\alpha n}) \leq 2^{a+1}\} \leq \pp \{\E_{l} Z (F_{\alpha n +l })  \leq 2^{a+1} (1-\frac{1}{2^k})^l  \}, $$
and defining $T := 2 (1-2^{-k})^l ,$ we have
\begin{eqnarray*}
\pp\{ Z(F_{\alpha n +l }) \leq 2^a \} &\geq& \pp\{ Z(F_{\alpha n +l }) \leq 2^a | \E_{l} Z (F_{\alpha n +l })  \leq T 2^{a} \} \pp\{ \E_{l} Z (F_{\alpha n +l })  \leq T 2^{a} \} \\
&\geq& (1-T) \pp\{ \E_{l} Z (F_{\alpha n +l })  \leq T 2^{a} \},
\end{eqnarray*}
where last inequality follows from Markov's inequality.
Thus,
$$ \pp \{ Z(F_{\alpha n}) < 2^{a+1} \} -\pp \{ Z(F_{\alpha n+l}) < 2^{a}\} \leq  2 (1-\frac{1}{2^k})^l.$$
By setting $l=A(k)\log(n)$, previous upper bound is $2 n^{A(k) \log (1-\frac{1}{2^k})}$, hence by taking $A(k)$ appropriately we get \eqref{l}.
%
\end{proof}

We are now in position to prove Lemma \ref{crit}.

\begin{proof}[Proof of Lemma \ref{crit}.]  Let $G$ be a satisfiable formula with a bounded number of clauses and, say, $r$ variables. And let us assume without loss of generality that $G$ contains only variables $x_1,\ldots,x_r$, and
that it is satisfied when all variables
are set to true. We then have
\begin{eqnarray*}
\pp \{ Z(F(n,\alpha_{\rm s})) < 2^{n \phi} | G\} 
&\leq& \pp \{ Z(F(n,\alpha_{\rm s})|_{x_1=\ldots=x_r=T} ) < 2^{n \phi} \} ,
\end{eqnarray*}
where $T$ refers to the true assignment. 
The random formula $F(n,\alpha_{\rm s})|_{x_1=\ldots=x_r=T}$ has now a different structure and probability distribution. Let us denote by $F_{*}$ an equivalently distributed random formula, which contains clauses of size 1 to $k$ and has only $n-r$ variables. 
Since $|C_k(n)| ={n \choose k} 2^k$, the number of $l$-clauses appearing in $F_*$, denoted by $S_l$, satisfies 
\begin{eqnarray}
\E S_k &=& \Theta(n),\\
\E S_{k-1} &=&  O(1),\label{O1}
\end{eqnarray}
and more generally 
\begin{eqnarray}
\E S_{k-i} = O(n^{1-i}),\quad \forall 1 \leq i \leq k-1\, . \label{O2}
\end{eqnarray}
Hence, defining a random formula $F_{**}$, which contains only $k$-clauses and $(k-1)$-clauses, with the $k$-clauses selected independently in $C_k(n-r)$ with probability $p=\frac{\alpha_{\rm s} n}{{n \choose k} 2^k}$, and exactly $d$ 
clauses of size $(k-1)$ uniformly selected in $C_{k-1}(n-r)$, 
we get that for any $\tau >0$ and $n$ large enough
\begin{eqnarray*}
\pp \{ Z(F(n,\alpha_{\rm s})) < 2^{n \phi} |  G\} 
&\leq& \pp \{ Z(F_{**}) < 2^{n \phi} \} +  \tau\, ,
\end{eqnarray*}
provided $d$ is large enough.
Since 1-clauses are more constraining than $(k-1)$-clauses, we can upper bound our estimate by replacing each $k$-clause by the disjunction of its $k$ 
literals. 
Moreover, the clauses on the $n-r$ variables are drawn with probability $p(n,\alpha_{\rm s})=\alpha_{\rm s} n/\{ {n \choose k} 2^k\}$, hence, by drawing them from $p(n-r,\alpha_{\rm s})$ instead we get, for 
$D=d(k-1)$,
\begin{align}
& \pp \{ Z(F(n,\alpha_{\rm s})) < 2^{n \phi} |  G\} 
\leq \pp \{ Z(F(n-r,\alpha_{\rm s})\wedge_{i=1}^D C^{(1)}_i)  < 2^{n \phi} \} +  \tau. \label{D}
\end{align}  
We now prove a useful fact.
\begin{remark}\label{rem1}
Let $f$ be a Boolean formula on $n$ variables, $C^{(l)}_i$ be $l$-clauses independently and uniformly selected in $C_l(n)$ and $y \in \mR$.
Then for any unbounded increasing sequence $h(n)$ and for any $\varepsilon>0$, we can take $n$ large enough such that,
$$ \pp \{ Z (f\wedge_{i=1}^D C^{(1)}_i) < y \} \leq \pp \{ Z (f\wedge_{i=1}^{h(n)} C^{(k)}_i )< k^D y \}  + \varepsilon.$$
\end{remark}
\begin{proof}
In order to prove this fact, we check that for $n$ large enough,
\begin{align}
\pp \{ Z ( f\wedge C^{(1)}_1 ) < y \} \leq \pp \{ Z ( f\wedge_{i=1}^{h(n)/D} C^{(k)}_i )< ky \}  + \varepsilon/D. \label{g}
\end{align}
Note that for any $g \geq 1$
\begin{align}
\pp \{ Z (f\wedge_{i=1}^{g} C^{(k)}_i ) < ky\} \geq \pp \{ \min_{1 \leq i \leq g}Z (f\wedge C^{(k)}_i )< k y \} \label{1} 
\end{align}
Moreover,
\begin{align}
\pp \{ Z ( f\wedge C^{(k)}_1) < ky \} &= \pp \{ Z (f\wedge \vee_{i=1}^{k} C^{(1)}_i )< ky \} \notag \\
 &\geq \pp \{ \sum_{i=1}^{k} Z ( f\wedge C^{(1)}_i )< ky \}\notag\\
 &\geq \pp \{ Z (f\wedge C^{(1)}_i) < y,\,\forall 1\leq i \leq k \}\notag\\
 &= \pp \{ Z (f\wedge C^{(1)}_1 )< y\}^k. \label{2} 
\end{align}
Therefore, putting \eqref{1} and \eqref{2} together, we get
\begin{eqnarray}
\pp \{ Z (f\wedge_{i=1}^{g} C^{(k)}_i) < ky\} &\geq&1- (1-  \pp \{ Z ( f\wedge C^{(1)}_1 )< y\}^k )^g, \notag
\end{eqnarray}
but this implies that we can take $g$ large enough, such that 
\begin{align}
 \pp \{ Z ( f\wedge C^{(1)}_1 )< y\} \leq \pp \{ Z( f\wedge_{i=1}^{g} C^{(k)}_i) < ky\} +\varepsilon/D, \notag
\end{align}
which proves the remark.
\end{proof}
 
We now can use our remark to upper bound the estimate in Eq.~\eqref{D}
\begin{align}
\pp \{ Z(F(n,\alpha_{\rm s})) < 2^{n \phi} |  G\} 
\leq \pp \{ Z(F(n-r,\alpha_{\rm s})\wedge_{i=1}^{h(n)} C^{(k)}_i)  < k^D 2^{n \phi} \} + 2\tau \label{t1}
\end{align}  
and by taking $t$ large enough 
\begin{align}
\pp \{ Z(F(n-r,\alpha_{\rm s})\wedge_{i=1}^{h(n)} C^{(k)}_i)  < k^D 2^{n \phi} \} \leq
\pp \{ Z(F(n-r,\alpha_{\rm s}+th(n)/(n-r))  < k^D 2^{n \phi} \}  +\tau. \label{t2}
\end{align}  
Defining $H(n):=th(n)/(n-r)$, we get from Lemma \ref{lip1} $(ii)$ 
\begin{align}
\pp \{ Z(F(n-r ,\alpha_{\rm s} +H(n) ) ) < k^D 2^{r \phi} 2^{(n -r) \phi} \}
\leq \pp \{ Z(F(n-r ,\alpha_{\rm s}+H(n)) )< 2^{(n-r)  \phi} \} +\tau. \label{t3}
\end{align}  
Note that by Lemma \ref{lip1} $(ii)$, the inequality in \eqref{t3} holds for a clause density which is independent of $n$, although here the clause density is $\alpha_{\rm s} (n)+H(n)$. Since we will pick $H(n)$ to be $o(1/\sqrt{n})$, hence the variation of $H(n)$ can be neglected. Regarding the variation of $\alpha_{\rm s} (n)$, note that this sequence fluctuates on a compact interval (for a fixed $\phi<1$, on an interval contained in $(0,A]$ where $A$ is an upper bound on the critical threshold of Friedgut \cite{Friedgut}), then one can check that the gap in this inequality, i.e., $P_n(\alpha,\phi) - P_n(\alpha,\phi + 1/n)$, tends to zero uniformly in $\alpha$, leading to the claimed inequality. 

Putting \eqref{t1}, \eqref{t2} and \eqref{t3} together, we get  
\begin{align}
\pp \{ Z(F(n,\alpha_{\rm s})) < 2^{n \phi} | G\}& \leq
\pp \{ Z(F(n-r ,\alpha_{\rm s}+H(n)) )< 2^{(n-r)  \phi} \} + 4 \tau \notag \\
& = P_{n-r}(\alpha_{\rm s}+H(n),\phi)  +4\tau \notag \\
&\leq P_n(\alpha_{\rm s}+2H(n),\phi)  +5\tau. \label{var}
\end{align}  
To see that the inequality \eqref{var} holds, let us check that for $\alpha,\phi$ fixed, $g(n)$ increasing and $n$ large enough
\begin{align}
P_{n-1}(\alpha,\phi) \leq P_n(\alpha + g(n)/n,\phi) + \tau.\label{var2}
\end{align}
The inequality \eqref{var} can then be verified by an appropriate choice of $g(n)$ and by using a similar argument as discussed previously regarding the dependence in $n$ of the clause density.
Recall that a random formula in $F(n-1,\alpha)$ is drawn by picking each clause in $C_k(n-1)$ with probability $p(n-1,\alpha)= \frac{\alpha (n-1)}{{n-1 \choose k}2^k}$. By a coupling argument, since $p(n-1,\alpha) > p(n,\alpha)$, one can equivalently draw a first formula $F_1$ by picking each clause in $C_k(n-1)$ with probability $p(n,\alpha)$ and a second formula $F_2$ by picking each clause in $C_k(n-1)$ with probability $p(n-1,\alpha)-p(n,\alpha)$; creating the formula $F_1 \wedge F_2$. Note that a random formula $F(n,\alpha)$ picks each clause in $C_k(n-1)$ and also in $\{C_k(n)- C : C \in C_k(n-1)\}$ with probability $p(\alpha,n)$. Hence
\begin{align}
& P_{n-1}(\alpha,\phi) =  \pp\{ Z(F(n-1,\alpha)) < 2^{(n-1) \phi}\}\notag \\
&= \pp\{ Z(F_1 \wedge F_2) < 2^{(n-1) \phi}\} \notag \\
& \leq \pp\{ Z(F(n,\alpha) \wedge F_2) < 2^{(n-1) \phi}\}. \label{up1}
\end{align}
The expected number of clauses in $F_2$ is given by $\E |F_2| = |C_k(n-1)| (p(n-1,\alpha)-p(n,\alpha)) = O(1)$, hence, we can upper bound \eqref{up1} by replacing $F_2$ with a constant number of random 1-clauses and use Remark \ref{rem1} (as done above) to conclude that
\begin{align}
& P_{n-1}(\alpha,\phi) \leq \pp\{ Z(F(n,\alpha + g(n)/n)) < K 2^{(n-1) \phi}\}
\end{align}
for a constant $K$ and an increasing function $g(n)$.
(Note that the 1-clauses are drawn within the set of $n-1$ variables instead of $n$ variables, but this does not change the conclusion). 
Finally, using Lemma \ref{lip1} $(ii)$, we get
\begin{align}
& P_{n-1}(\alpha,\phi)  \leq \pp\{ Z(F(n,\alpha + g(n)/n) < 2^{n \phi}\} + \tau = P_{n}(\alpha + g(n)/n,\phi) + \tau,
\end{align}
which proves \eqref{var2}. Hence, we have
\begin{align*}
\pp \{ Z(F(n,\alpha_{\rm s})) < 2^{n \phi} | G\}
&\leq P_n(\alpha_{\rm s}+2H(n),\phi)  +5\tau 
\end{align*}  
and by choosing $h(n)=o(\sqrt{n})$ (increasing), i.e., $H(n)=o(1/\sqrt{n})$, and using Lemma \ref{lip1} $(i)$, we get 
\begin{align*}
\pp \{ Z(F(n,\alpha_{\rm s})) < 2^{n \phi} |  G\}  \leq P_n(\alpha_{\rm s},\phi ) +6 \tau. 
\end{align*}  
\end{proof}

\subsection{Proofs of Theorem \ref{convpsat} and Lemma \ref{guerra}}\label{proof2}

\begin{proof}[Proof of Lemma \ref{guerra}.] We refer to the proof of Lemma \ref{guerracsp} which is more general. 
\end{proof}

In order to prove Theorem \ref{convpsat}, we first need the following technical lemma.
%
\begin{lemma}\label{anal}
Let $\Delta(n) = O(n/(\log n)^{1+\e})$ for some $\e>0$, and $t(n)=o(n)$.
Let $f(\cdot)$ be positive, such that $f(n)/n$ is bounded above and
$$f(n_1 + n_2) + \Delta(n_1 + n_2 ) \geq  f(n_1) + f(n_2) ,\quad \forall n_1, n_2 \geq t(n_1+n_2).$$
Then $f(n)/n$ converges.
\end{lemma}

\noindent
{\bf Remark 3.} This lemma still holds if $\Delta(n)$ is such that $\sum_{n} \frac{\Delta(n)}{n^2} < \infty$. 
\begin{proof}
Let $\e >0$.
Since $f(n)/n$ is bounded above, we have $S := \lim \sup_n f(n)/n < \infty$.
Let $n_0$ large enough such that $f(n)/n < S + \e$ for any $n \geq n_0$.
Let $r$ be an integer (to be chosen at our convenience but to be kept fixed) and let $\gamma(n)$ such that $\gamma(n) \geq t(n)$ and $\gamma(n)=o(n)$. Define $$m(n) = \inf \{ r 2^k : r 2^k \geq \gamma (n) \}$$
and note that $ \gamma(n) \leq m(n) \leq 2 \gamma(n) \vee r$.
For any $n$, we then have 
$$n= \lfloor n/m(n) -1 \rfloor m(n) + q(n),$$
where $q(n) \in [m(n), 2m(n)]$, with $\lfloor n/m(n) -1 \rfloor \geq 0$ for $n$ large enough.
Hence, using the property of $f$, we have
\begin{align*}
f(n) \geq  \lfloor n/m(n) -1 \rfloor f(m(n)) + f(q(n)) - \sum_{k=1}^{\lfloor n/m(n) -1 \rfloor} \Delta (q(n)+ k m(n)), 
\end{align*}
and since $f$ is positive,
\begin{align}
\frac{f(n)}{n} \geq  \frac{m(n)}{n} \lfloor n/m(n) -1 \rfloor \frac{f(m(n))}{m(n)}   -\frac{1}{n} \sum_{k=1}^{\lfloor n/m(n) -1 \rfloor} \Delta (q(n)+ k m(n)). \label{tot}
\end{align}
Since $\gamma (n) = o(n)$,  for $n$ large enough, we have $$\frac{m(n)}{n} \lfloor n/m(n) -1 \rfloor > 1-\e.$$
We now show that for $n$ large enough, we also have $\frac{f(m(n))}{m(n)} > S - \e$. First, note that we can take $r$ large enough such that $r 2^k / 2 \geq \gamma (r 2^k)$ for all $k \geq 0$, since $\gamma (n) = o(n)$.
Hence, using the property of $f$, we have 
$$\frac{f(r 2^{k+1})}{r 2^{k+1}} \geq  \frac{f(r 2^k )}{r 2^{k}} - \frac{1}{r 2^{k+1}} \Delta (r 2^{k+1}) $$
and 
\begin{align}
 \frac{f(r 2^{k})}{r 2^{k}} \geq  \frac{f(r  )}{r } - \sum_{j=1}^k \frac{1}{r 2^{j}} \Delta (r 2^{j}) . \label{ugly}
\end{align}
Since we can pick $r$ at our convenience, note that if $r$ is a power of 2, 
$$\sum_{j=1}^k \frac{1}{r 2^{j}} \Delta (r 2^{j}) = \sum_{j = \log_2 r+1}^{\log_2 r +k} \frac{1}{2^{j}} \Delta ( 2^{j}),$$
which is, when $r$ increases, tending to zero uniformly in $k$, provided that 
$$\sum_{j=1}^\infty \frac{1}{ 2^{j}} \Delta ( 2^{j}) < \infty.$$
Since previous condition follows from our hypothesis on $\Delta$,
and since we can always take $r$ large enough to ensure that $f(r)/r > S-\e$,  
we can take $r$ large enough such that, from \eqref{ugly}, we have for any $k$  
$$ \frac{f(r 2^{k})}{r 2^{k}} \geq  S -  \e$$
and since $m(n)$ is of the form $r2^k$, for any $n$
$$ \frac{f(m(n))}{m(n)} \geq  S -  \e.$$
Finally, we need to show that the last term in \eqref{tot} is vanishing, i.e., that
$$\frac{1}{n} \sum_{k=1}^{\lfloor n/m(n) -1 \rfloor} \Delta (q(n)+ k m(n)) \stackrel{n\rightarrow \infty}{\longrightarrow} 0.$$
For this, we pick $\gamma(n)$ to be large enough. For example, if $\Delta = O(n/(\log n)^{1+\e})$, we have
\begin{align}
\frac{1}{n} \sum_{k=1}^{\lfloor n/m(n) -1 \rfloor} \Delta (q(n)+ k m(n)) \leq \frac{1}{n} \frac{n}{m(n)} \Delta (n),
\end{align} 
and since $m(n) \geq \gamma (n)$, if $\gamma (n)= O(n/(\log n)^{1+\nu})$ with $\nu > \e$, we conclude the proof. In general, we pick $\gamma(n)$ such that $ \Delta (n)/\gamma(n)=o(1)$. 
\end{proof}

\begin{proof}[Proof of Theorem \ref{convpsat}.] 
Let $F_n = F(n,\alpha)$.
Note that
\begin{align*}
& \E \log (1 + Z( F_{n})) = \E [\log (1 + Z( F_{n})) ,Z (F_{n}) \geq 1] 
\end{align*}
and
\begin{align*}
&\E [\log (1 + Z( F_{n})) ,Z (F_{n}) \geq 1] \\
& = \E [\log Z (F_{n}), Z( F_{n} )\geq 1]+\E [\log (1+Z(F_n)^{-1} ) , Z (F_{n}) \geq 1].
\end{align*}
Let $c>0$, we have
\begin{align*}
& \E [\log (1+Z(F_n)^{-1} ) , Z (F_{n}) \geq 1] \leq \E [Z(F_n)^{-1}  , Z (F_{n}) \geq 1] \\
 & \leq \E [Z(F_n)^{-1}  , Z (F_{n} )\geq 1, Z_{FV} \geq c n  ] + \pp\{ Z_{FV} < c n ,   Z (F_{n}) \geq 1 \} 
\end{align*}
where $Z_{FV}$ is the number of free variables in $F_n$. Therefore, 
$\E [Z(F_n)^{-1}  , Z (F_{n} )\geq 1, Z_{FV} \geq c n  ] \leq e^{-cn}$. Moreover, there exists $c, c_2>0$ such that
$$ \pp\{ Z_{FV} < c n, Z(F_{n}) \geq 1  \} =O(e^{-c_2 n}) $$
hence there exists $\xi>0$ such that
\begin{align}
& \tau(n) := \E [\log (1+Z(F_n)^{-1} ) , Z (F_{n}) \geq 1] = O(e^{- \xi n}). \label{claim}
\end{align} 
On the other hand, we have
(denoting by $F_{n_1}$, $F_{n_2}$ two independent formulas and letting $n=n_1+n_2$) 
\begin{align*}
& \E \log (1 + Z( F_{n_1}) Z(F_{n_2})) =\E[ \log(1+ Z (F_{n_1}) Z(F_{n_2})) , Z (F_{n_1}) Z(F_{n_2}) \geq 1]
\end{align*}
and
\begin{align*}
&\E[ \log(1+ Z (F_{n_1}) Z(F_{n_2})) , Z( F_{n_1})Z(F_{n_2}) \geq 1] \geq \E[ \log( Z (F_{n_1})Z(F_{n_2})) , Z (F_{n_1}) Z(F_{n_2}) \geq 1].
\end{align*}
Hence, using Lemma \ref{guerra}, we get the following inequality
\begin{align*}
& \E [\log  Z (F_{n}), Z (F_{n}) \geq 1]+ \tau(n) \geq \E[ \log( Z( F_{n_1}) Z(F_{n_2})) , Z (F_{n_1}) Z(F_{n_2}) \geq 1]
\end{align*}
or equivalently 
\begin{align}
& g(n)+ \tau(n) \geq g(n_1) + g(n_2) - g(n_1) \e (n_2) - g(n_2) \e (n_1) , \quad \forall n_1,n_2 \geq k \label{start}
\end{align}
where 
\begin{eqnarray*}
 g(n) &=& \E [\log Z (F_{n}), Z (F_{n}) \geq 1], \\
  \e (n) &=& \pp \{Z( F_n) = 0 \}.
\end{eqnarray*} 
Note that $0 \leq g(n) \leq n$. Therefore \eqref{start} implies
\begin{align}
& g(n)+ \tau(n) \geq g(n_1) + g(n_2) - n_1 \e (n_2) - n_2 \e (n_1), \quad \forall n_1,n_2 \geq k. \label{gb}
\end{align}
Since $\alpha < \alpha^*$, we have that $\e(n) = O(1/(\log n)^{1+\e})$, for some $\e >0$.
We then restrict ourself to $$n_1,n_2 \geq t(n) := n/(\log n)^\eta, \quad \text{with } \eta = \e/3.$$ 
This implies that $n^{1-\eta} \leq n_2$ and $(1-\eta) \log n \leq \log n_2$. 
So, for $n_1, n_2$ large enough, we have
\begin{align}
& n_1 \leq n_2 (\log n_2)^{2 \eta} \\
& n_2 \leq n_1 (\log n_1)^{2 \eta}.
\end{align}
Going back to \eqref{gb}, we get
\begin{align*}
& g(n)+ \tau(n) \geq g(n_1) + g(n_2) - n_1 (\log n_1)^{2 \eta} \e (n_1) - n_2 (\log n_2)^{2 \eta}  \e (n_2), \quad \forall n_1,n_2 \geq n/(\log n)^{2 \eta} 
\end{align*}
or equivalently
\begin{align}
& f(n)+\Delta(n) \geq f(n_1) + f(n_2) , \quad \forall n_1,n_2 \geq t(n)
\end{align}
where
\begin{eqnarray*}
f(n)&=&g(n) -n (\log n)^{2 \eta} \e (n) \\
\Delta(n)&=&n (\log n)^{2 \eta} \e (n)+ \tau(n)\\
t(n)&=& n/(\log n)^{\eta}.
\end{eqnarray*}
But $ \e - 2 \eta = \e/3 >0$, hence $$\Delta (n) \leq O\left(\frac{n}{(\log n)^{1 + \e/3}}\right),$$ 
and we satisfy the hypothesis of Lemma \ref{anal}, which implies that $f(n)/n$ converges, hence $g(n)/n$ converges too. 
\end{proof}

\subsection{Proofs of Lemma \ref{conva} and Theorem \ref{res}}\label{proof3}

\begin{proof}[Proof of Lemma \ref{conva}.] 
From Theorem \ref{convpsat}, 
for every $\alpha < \alpha^*$, $\psi_n(\alpha)$ converges to a limit $\phi_s(\alpha)$. 
Note that $\phi_s(\,\cdot\,)$ is a non-increasing function on $\A$, hence, it has a countable number of plateaus and discontinuities.  
Let $\alpha_0 \in \A$ 
and denote $\phi_0 = \phi_s(\alpha_0)$.
If $\alpha_n (\phi_0)$ does not converge, define  
$$\underline{\alpha}_0=\lim \inf_{n \rightarrow \infty} \alpha_n (\phi_0)  ,$$
$$n_k \nearrow \infty \text{ s.t. } \lim_{k \rightarrow \infty} \alpha_{n_k} (\phi_0) = \underline{\alpha}_0,$$
$$\bar{\alpha}_0=\lim \sup_{n \rightarrow \infty} \alpha_n (\phi_0) ,$$
$$m_k \nearrow \infty \text{ s.t. } \lim_{k \rightarrow \infty} \alpha_{m_k} (\phi_0) = \bar{\alpha}_0.$$
Then, for any $\alpha \in (\underline{\alpha}_0,\bar{\alpha}_0)$, there exists $\e>0$ such that 
\begin{align}
P_{m_k}(\alpha,\phi_0) \leq P_{m_k} ( \alpha_{m_k} (\phi_0)-\e, \phi_0)
\stackrel{k \nearrow \infty}{\rightarrow} 0 \label{lim1}
\end{align}
and
\begin{align}
P_{n_k}(\alpha, \phi_0) \geq P_{n_k} (\alpha_{n_k} (\phi_0)+\e, \phi_0)
\stackrel{k \nearrow \infty}{\rightarrow} 1 \label{lim2}
\end{align}
Moreover, if $\alpha < \alpha^*$,
$$
P_{n_k}(\alpha, \phi_0)
= \pp \{ \frac{1}{n_k} \log Z(F(n_k,\alpha)) < \phi_0 | Z(F(n_k,\alpha)) \geq 1 \} + o(1),$$
hence 
$$\lim_{k \rightarrow \infty} \E [ \frac{1}{n_k} \log Z(F(n_k,\alpha))  | Z(F(n_k,\alpha)) \geq 1 ] \le \phi_0, $$
i.e., since $\psi_n(\alpha)$ converges to $\phi_s(\alpha)$ from Lemma \ref{convpsat}, $$\phi_s (\alpha) \le \phi_0.$$ Similarly, we have 
$$\lim_{k \rightarrow \infty} \E [ \frac{1}{m_k} \log Z(F(m_k,\alpha))  | Z(F(m_k,\alpha)) \geq 1 ] \ge \phi_0, $$
and
$$\phi_s (\alpha) \ge \phi_0.$$ Therefore, $\phi_0$ is a plateau of $\phi_s(\cdot)$, and since $\phi_s(\cdot)$ has countably many plateaus, there are countably many $\phi_0 \in \phi_s(\A)$, for which $\alpha_n(\phi_0)$ does not converge.  
\end{proof}

\begin{proof}[Proof of Theorem \ref{res}]
From Theorem \ref{convpsat}, there exists a function $\phi_s(\cdot)$, such that for any $\alpha \in [0,\alpha^*)$, we have $\f(\alpha) = \lim_{n \rightarrow \infty} \psi_n(\alpha)$, where $\psi_n(\cdot)$ is defined in Theorem \ref{convpsat}. Let $I:=\f ([0,\alpha^*))$. From Lemma \ref{conva}, there exists a countable set $\mathcal{C} \subseteq I$ and a function $A: \,I \setminus \mathcal{C} \rightarrow [0,\alpha^*)$ such that for any $\phi \in I \setminus \mathcal{C}$, we can define the limit $A(\phi)=\lim_{n \rightarrow \infty} \alpha_n (\phi)$. Note that for any $\phi \in I \setminus \mathcal{C}$, Theorem \ref{fried} implies $\f(A(\phi))=\phi$. 

Now, for any $\alpha \in [0,\alpha^*)$ which is not a discontinuity point of $\f$ (this holds except on a countable subset of $[0,\alpha^*)$), 
and for any $\e >0$, there exists $\e^\prime <\e$ such that $\phi_* :=\f(\alpha) - \e^\prime \in I \setminus \mathcal{C}$ and hence $\alpha_n(\phi_*)$ tends to a limit $A_*$. Note that $A_*>\alpha$, since $\alpha$ is not a discontinuity point of $\f$ and since $\f(A_*) = \phi_*$. Therefore, there exists $\delta >0$ such that 
$$P_n(\alpha, \phi_s (\alpha) - \e) \leq P_n(\alpha, \phi_*) \leq P_n(\alpha_n(\phi_*) - \delta, \phi_*)$$
and we conclude from Theorem \ref{fried} that $P_n(\alpha, \phi_s (\alpha) - \e) \rightarrow 0$ when $n \rightarrow \infty$. With a similar argument, we conclude that $P_n(\alpha, \phi_s (\alpha) + \e) \rightarrow 1$ when $n \rightarrow \infty$. 
\end{proof}

\subsection{Proofs of Lemma \ref{ch1}, Lemma \ref{ch2} and Lemma \ref{guerracsp}}\label{proof4}

\begin{proof}[Proof of Lemma \ref{guerracsp}.] In this proof, we keep $\alpha$ fixed and split the $n$ variables into two sets of $n_1$ and $n_2=n-n_1$ variables, such as $\{1,\ldots,n_1\}$ and $\{n_1+1,\ldots,n\}$.
For convenience, we now work with the interpolated Poisson model.
We construct a random Boolean formula as follows: we first draw independently the integers $M$, $M_1$ and $M_2$ under Poisson distributions of parameters $\alpha n t$, $\alpha n_1 (1-t)$ and $\alpha n_2 (1-t)$ respectively. We then draw independently $M$ clauses from the full system, by picking for each clause the indices of the variables appearing in it independently and uniformly at random within the set of $n$ variables and by picking $\vp$ under $\mu$. We also draw independently $M_i$ clauses from each sub-systems, by picking for each clause the indices of the variables appearing in it independently and uniformly at random within the set of $n_i$ variables and by picking $\vp_i$ under $\mu$. Finally, we take the conjunction of all clauses to create the formula $F_n (t)$. 

Note that the claim of the lemma is equivalent to
\begin{align}
\E \log (1 + Z (F_n(1))) \geq \E \log (1 + Z( F_{n}(0))) , \label{qsub}
\end{align}
which is proved by showing that
$$\frac{\de\phantom{t}}{\de t} \E \log (1 + Z (F_n(t))) \geq 0.$$
An elementary calculation yields
\begin{align}
& \frac{\de\phantom{t}}{\de t} \frac{1}{n} \E \log (1 + Z (F_n(t))) = \alpha  \left[\E_{\vp,I} \E \log (1 + Z (F_n(t) \wedge \vp(x_I))) - \E \log (1 + Z (F_n(t))) \right]  \notag \\
& - \alpha \frac{n_1}{n} \left[\E_{\vp_1,I_1} \E \log (1 + Z (F_n(t) \wedge \vp_1 (x_{I_1}))) - \E \log (1 + Z (F_n(t))) \right] \notag \\& - \alpha \frac{n_2}{n} \left[\E_{\vp_2,I_2} \E \log (1 + Z (F_n(t) \wedge \vp_2(x_{I_2}) )) - \E \log (1 + Z (F_n(t))) \right], \notag 
\end{align}
where $\vp, \vp_1, \vp_2 \iid \mu$, $I \sim U^k$, $I_1 \sim U_1^k$, $I_2 \sim U_2^k$, all independent, and where $U^k$, respectively $U_i^k$, denotes the $k$-th product measure of $U$, respectively $U_i$ (where $U$, resp. $U_i$, denotes the uniform measure on the $n$ variables, resp. $n_i$ variables). Hence, $x_I = (x_{i_1}, \ldots, x_{i_k})$ with $i_1,\ldots, i_k$ iid uniform over the $n$ variables.

We then have
\begin{align}
&\E_{\vp, I} \E \log (1 + Z (F_n(t) \wedge \vp(x_I))) - \E \log (1 + Z (F_n(t)))    = \E_{\vp, I} \E \log \langle   \vp(X_I) \rangle  \notag
\end{align}
where $X$ is uniformly drawn within the augmented solution space $S(F_{n^*}(t))= S(F_{n}(t))\cup \{*\}$, where $*$ is an assignment which returns true on any Boolean functions, and  $\langle\, \cdot\, \rangle$ denotes the expectation with respect to $X$. Note that 
\begin{align}
\E_{\vp, I} \E \log \langle   \vp(X_I) \rangle  = - \E_{\vp, I} \E \sum_{l=1}^\infty  \frac{ \langle \tilde{\vp}(X_I) \rangle^l }{l} . \label{serie}
\end{align}
where $\tilde{\vp} = 1 - \vp$. 
We now introduce the `replicas' $X^{(r)}$, which are independent and
identically distributed copies of $X$. We then have
\begin{align*}
& \langle  \tilde{\vp}(X_I) \rangle^l = \langle \prod_{r=1}^l \tilde{\vp}(X^{(r)}_I) \rangle, \quad \forall l \geq 1. 
\end{align*}
We are done if we can show that for any realizations of the $X^{(r)}$'s and for any $l \geq 1$, 
\begin{align}
& \E_{\vp, I} \prod_{r=1}^l \tilde{\vp}(X^{(r)}_I) - \frac{n_1}{n} \E_{\vp, I_1} \prod_{r=1}^l \tilde{\vp}(X^{(r)}_{I_1}) - \frac{n_2}{n} \E_{\vp, I_2} \prod_{r=1}^l \tilde{\vp}(X^{(r)}_{I_2}) \geq 0. \label{targ}
\end{align}
Note that 
\begin{align}
\E_{\vp, I} \prod_{r=1}^l \tilde{\vp}(X^{(r)}_I) = \E_\vp \E_{\hat{P}} \prod_{r=1}^l \tilde{\vp}(\xi^{(r)}) 
\end{align}
where $\xi^{(1)}, \ldots, \xi^{(l)} \iid \hat{P}$ and where $\hat{P}$ is the empirical distribution of $X^{(1)}_I, \ldots, X^{(l)}_I$, i.e. the distribution on $\{-1,1\}^{kl}$ given by
$$\hat{P}(x_1^{(1)},\ldots, x_k^{(1)}, \ldots, x_1^{(l)}, \ldots, x_k^{(l)})=\prod_{i=1}^k\bar{P}(x_i^{(1)},\ldots, x_i^{(l)}) $$
with 
$$ \bar{P}(x_i^{(1)},\ldots, x_i^{(l)}) = \frac{\# \{ i \in \{1,\ldots,n\}: (X_{i}^{(1)},\ldots, X_{i}^{(l)})=(x_i^{(1)},\ldots, x_i^{(l)}) \}}{n}$$
and similarly
\begin{align}
\E_{\vp, I_s} \prod_{r=1}^l \tilde{\vp}(X^{(r)}_{I_i}) = \E_\vp \E_{\hat{P}_s} \prod_{r=1}^l \tilde{\vp}(\xi_s^{(r)}), \quad s =1,2 
\end{align}
where $\xi_s^{(1)}, \ldots, \xi_s^{(l)} \iid \hat{P}_s$ and where $\hat{P}_s$ is the empirical distribution of $X^{(1)}_{I_s}, \ldots, X^{(l)}_{I_s}$, i.e. the distribution on $\{-1,1\}^{kl}$ given by
$$\hat{P}_s(x_1^{(1)},\ldots, x_k^{(1)}, \ldots, x_1^{(l)}, \ldots, x_k^{(l)})=\prod_{i=1}^k\bar{P}_s(x_i^{(1)},\ldots, x_i^{(l)}), \quad s=1,2 $$
with 
\begin{align*}
& \bar{P}_1(x_i^{(1)},\ldots, x_i^{(l)}) = \frac{\# \{ i \in \{1,\ldots,n_1\}: (X_{i}^{(1)},\ldots, X_{i}^{(l)})=(x_i^{(1)},\ldots, x_i^{(l)}) \}}{n_1}, \\
& \bar{P}_2(x_i^{(1)},\ldots, x_i^{(l)}) = \frac{\# \{ i \in \{n_1+1,\ldots,n\}: (X_{i}^{(1)},\ldots, X_{i}^{(l)})=(x_i^{(1)},\ldots, x_i^{(l)}) \}}{n_2}.
 \end{align*}
Now, using the operator $\Gamma_l$ defined by \eqref{cond1} in H2, i.e., 
\begin{align*}
\Gamma_l: \, M_1(\{-1,1\}^l) \ni \nu \mapsto \E_\vp \E_{Z^{(r)}} \prod_{r=1}^l (1-\vp(Z^{(r)})) 
\end{align*}
where $Z^{(r)}$ are Boolean random vectors of dimension $k$ such that $Z_i=(Z_i^{(1)},\ldots, Z_i^{(l)})$, $i=1,\ldots,k$, are i.i.d. with distribution $\nu$, note that \eqref{targ} is equivalent to
$$\Gamma_l (\bar{P}) - \frac{n_1}{n} \Gamma_l (\bar{P}_1)- \frac{n_2}{n} \Gamma_l (\bar{P}_2) \geq 0,$$
which holds by convexity of $\Gamma_l$, since 
$$ \bar{P} = \frac{n_1}{n} \bar{P}_1 + \frac{n_2}{n} \bar{P}_2.$$
\end{proof}

\begin{proof}[Proof of Lemma \ref{ch1}.]
We have
\begin{align*}
\Gamma_l ( \nu ) &=\E_\vp \E_{Z^{(r)}} \prod_{r=1}^l \tilde{\vp}(Z^{(r)}) \\
&=\E_{\vp} \sum_{z^{(1)}, \ldots, z^{(l)} \in \{-1,1\}^k}  \prod_{r=1}^l \tilde{\vp}(z^{(r)}) \nu^k (z^{(1)},\ldots,z^{(l)})\\
& =\E_{s^{(j)}} \sum_{z^{(1)}, \ldots, z^{(l)} \in \{s^{(1)}, \ldots, s^{(J)}\} }  \prod_{i=1}^k \nu ((z^{(1)})_i,\ldots,(z^{(l)})_i) \\
& = \sum_{i_1,\ldots, i_l \in \{1,\ldots, J\}} [\E_{s^{(j)}_1} \nu(s_1^{(i_1)}, \ldots, s_1^{(i_l)})]^k 
\end{align*}
and $\Gamma_l$ is convex for any $l\geq 1$. 
\end{proof}

\begin{proof}[Proof of Lemma \ref{ch2}.]
We need to check the convexity of 
\begin{align}
\nu \mapsto \E_s \E_{Z^{(r)}} \prod_{r=1}^l (1-\vp_s(Z^{(r)})) 
\end{align}
where $Z^{(r)}$ are Boolean random vectors of dimension $k$ such that $Z_i=(Z_i^{(1)},\ldots, Z_i^{(l)})$, $i=1,\ldots,k$, are i.i.d. with distribution $\nu$, 
$$\vp_s (x) = \1 (\prod_{i=1}^k x_i = s)$$
and 
$$\mu(\vp_1) = \mu(\vp_{-1})=1/2.$$
Note that
\begin{align}
\E_{Z^{(r)}} \prod_{r=1}^l (1-\vp_s(Z^{(r)})) = \pp \{ \prod_{i=1}^k Z_i = -s^l \} 
\end{align}
where $-s^l$ denotes the vector $(-s,\ldots,-s)$ with $l$ components and where $ \prod_{i=1}^k Z_i$ denotes the component-wise product of the vectors $Z_i$. 
Since the $Z_i$ are i.i.d. under $\nu$ and valued in $\{-1,1\}$, and since we are interested in their product, we now work with their Fourier transform. 
For any $Q \subseteq \{1,\ldots, l\}$, let
$$f(Q)=f_{Z_1}(Q) =  \E \prod_{r \in Q} Z_1^{(r)}.$$
Note that we can recover the distribution of $Z_1$ by knowing $f(Q)$ for any $Q$, in particular
$$\pp \{Z_1 = 1^l\} = \sum_{Q \in 2^{[l]}} f(Q)$$
and
$$\pp \{Z_1 = -1^l\} = \sum_{Q \in 2^{[l]}} (-1)^{|Q|} f(Q).$$
Moreover,
$$f_{\prod_{i=1}^k Z_i}(Q) =  f(Q)^k,$$
hence,
\begin{align*}
\E_s \pp \{ \prod_{i=1}^k Z_i = -s^l \} &= 1/2 \sum_{Q \in 2^{[l]}} f(Q)^k +1/2 \sum_{Q \in 2^{[l]}} (-1)^{|Q|} f(Q)^k \\  
& = \sum_{Q \in 2^{[l]} \atop{|Q| \text{ even}}}  f(Q)^k.
\end{align*}
Since $f(Q)$ is linear in $\nu$ (it is the expectation of $\prod_{r \in Q} Z_1^{(r)}$ where $(Z_1^{(1)},\ldots,Z_1^{(l)}) \sim \nu$), the above summation is clearly convex in $\nu$ if $k$ is even.
\end{proof}

%
\bibliographystyle{amsalpha}

\bibliography{gen}
\end{document}